\documentclass[graybox]{svmult}

\usepackage{amsfonts}
\usepackage{amssymb}
\usepackage{amsmath}
\usepackage{amscd}      % adapts the commutative diagram macros of AMS-\TeX{}
\usepackage{fancyhdr}  

\usepackage{mathptmx}       % selects Times Roman as basic font
\usepackage{helvet}         % selects Helvetica as sans-serif font
\usepackage{courier}        % selects Courier as typewriter font
\usepackage{type1cm}        % activate if the above 3 fonts are
\usepackage{makeidx}         % allows index generation
\usepackage{graphicx}        % standard LaTeX graphics tool
\usepackage{multicol}        % used for the two-column index
\usepackage[bottom]{footmisc}% places footnotes at page bottom
\usepackage[vcentermath]{youngtab}

\newcommand{\ba}{\begin{array}}
\newcommand{\ea}{\end{array}}
\newcommand{\beq}{\begin{equation}}
\newcommand{\eeq}{\end{equation}}
\newcommand{\beqa}{\begin{eqnarray}}
\newcommand{\eeqa}{\end{eqnarray}}

\def\E{{\mathcal E}}

\def\C{{\mathbb C}}
\def\Z{{\mathbb Z}}
\def\g{{\mathfrak g}}
\def\h{{\mathfrak h}}
\def\n{{\mathfrak n}}
\def\p{{\mathfrak p}}
\def\f{{\mathcal V(p)}}

\def \rvac {\!\! \left. \left. \right| \! 0 \right\rangle}
\def\demi {\frac{1}{2}}

\begin{document}

\title*{Parafermionic algebras, their modules and cohomologies}
\author{Todor Popov}
\institute{Todor Popov \at INRNE, Bulgarian Academy of Sciences,
72 Tsarigradsko chauss\'ee, 1784 Sofia, Bulgaria \\ \email{tpopov@inrne.bas.bg}
%\and Name of Second Author \at Name, Address of Institute \email{name@email.address}
}
\maketitle

\abstract{ We explore the Fock spaces of the parafermionic algebra %generated by the creation and annihilation operators 
introduced by H.S. Green. Each parafermionic Fock space allows for a free minimal resolution by graded modules of the graded 2-step nilpotent subalgebra of the parafermionic creation operators. Such a free resolution is constructed with the help of a classical Kostant's theorem computing  Lie algebra cohomologies of the  nilpotent subalgebra with values in the parafermionic Fock space. %By calculation
The Euler-Poincar\'e characteristics of the parafermionic Fock space free resolution yields some interesting  identities between Schur polynomials.
Finally we briefly comment on parabosonic and general parastatistics Fock spaces.}

\section{Introduction}
\label{sec:1}

The parafermionic and parabosonic algebras were introduced by H.S. Green
as a inhomogeneous cubic algebra having as quotients  the fermionic and bosonic algebras with canonical (anti)commutation relations.
In an attempt to find a new paradigm for quantization of classical fields
H.S. Green introduced the parabosonic and parafermionic algebras
encompassing the bosonic and fermionic algebras based on the canonical
quantization scheme.
Here we are dealing with the Fock spaces of the  parafermionic algebra $\g$  of creation and annihilation operators.
These Fock spaces are particular parafermionic algebra modules build at
the top of a unique vacuum state by the creation operators. The creation operators close a free graded 2-step nilpotent algebra $\n$, %further referred as a {\it subalgebra} 
$\n \subset \g$.
The Fock space  of a parafermionic algebra $\g$ is then defined 
as a quotient module of the  free $\n$-module, where the quotient ideal 
stems from the generalization of the Pauli exclusion principle.
In this note we calculate the cohomologies $H^{\bullet}(\n, \f)$
of the nilpotent subalgebra $\n$ with coefficients in the parafermionic  Fock space $\f$ (taken as a $\n$-module).
The cohomology ring $H^{\bullet}(\n, \f)$ is obtained due to by now classical  Kostant's theorem \cite{Kostant}.
With the data of $H^{\bullet}(\n, \f)$ one is able to construct a 
minimal resolution by free $\n$-module of the Fock space $\f$.
Its existence is garanteed by the Henri Cartan's results on graded algebras. It turns out that the Schur polynomials  identities which have been recently put forward \cite{LSVdJ,SVdJ} by Neli Stoilova and Joris Van der Jeugt stem from the  Euler-Poincar\'e characteristics of the minimal free resolutions of the parafermionic and parabosonic Fock space.

\section{Parafermionic and parabosonic algebras}
\label{sec:2}

The parafermionic algebra $\g$ with finite number $n$ degrees of freedom
is a Lie algebra with a Lie bracket $[ \bullet , \bullet ]$
generated by the creation $a_i^{\dagger}$ and annihilation $a^j$
operators ($i,j=1,\ldots, n$)  having the following exchange relations
 
% Always give a unique label
% and use \ref{<label>} for cross-references
% and \cite{<label>} for bibliographic references
% use \sectionmark{}
% to alter or adjust the section heading in the running head

\beqa 
\ba{rcccrcc} 
[[ a^{\dagger}_{i},a^{j} ], a^{\dagger}_{k}] &=& 2 \delta_{k}^j 
a^{\dagger}_{i} \ ,&\quad
& [[ a^{\dagger}_{i},a^{j} ], a^{k}] &=& - 2 \delta_{i}^k a^{j} \ ,
\\[4pt] [ [a^{\dagger}_{i}, a^{\dagger}_{j}],a^{\dagger}_{k} ] &=&0 \ , & \quad &[
[a^{i}, a^{j}],a^{k} ] &=&0  \ .\ea \label{PCR} 
\label{1}
\eeqa
The parafermionic algebra $\g$ with finite number  degrees of freedom $n$ is isomorphic to  the  semi-simple Lie algebra 
\beq
\label{bn}
\g= \h \oplus \bigoplus_{\alpha\in \Delta_+} \g_{\alpha} 
\oplus \bigoplus_{\alpha\in \Delta_-} \g_{\alpha}  \ ,
 \eeq
for a root system $\Delta= \Delta_+\cup \Delta_-$ of type $B_n$ with positive roots $\Delta_+$ given by
$$\Delta_+= \{ e_i \}_{1\leq i \leq n} \cup \{e_i + e_j, e_i-e_j\}_{1\leq i<j\leq n} \ , \quad \mbox{and} \quad \Delta_-=-\Delta_+ \ .$$
Here $\{ e_i \}_{i=1}^n$ stand for the orthogonal basis in the root space,
$(e_i| e_j)=\delta_{ij}$.
One concludes that the parafermionic algebra $\g$ with $n$ degrees of freedom is isomorphic to the orthogonal algebra $\g \cong\mathfrak{so}_{2n+1}$ endowed with the anti-involution $\dagger$. 
The physical generators correspond to
the Cartan-Weyl basis $a^{\dagger}_i:=E^{e_i}$ and $a^j:=E^{-e_j}$ .

%It was shortly realized \cite{} that the latter parabosonic algebra   is isomorphic to the orthogonal Lie algebra 
%of Dynkin type $B_n$, i.e., the orthosymplectic algebra $\mathfrak {so}_{2n+1}$.
Similarly one defines the parabosonic algebra $\tilde{\g}$
with exchange relations (\ref{1}) as the Lie super-algebra
 endowed with a Lie super-bracket $[\bullet , \bullet ]$ whose generators $a^{\dagger}_i$ and $a^j$ are taken to be odd generators. The parabosonic algebra $\tilde{\g}$
with $m$ degrees of freedom
 is shown \cite{GP} to be isomorphic to the Lie super algebra of type $B_{0,m}$ in the Kac table, i.e., $\mathfrak{osp}_{1|2m}$. More generally,
 one defines the parastatistics algebra as the Lie super-algebra with $n$ even parafermionic  and $m$ odd parabosonic degrees of freedom.
 The parastatistics algebra is shown to be isomorphic to the
 super-algebra of type $B_{n,m}$, i.e., $\mathfrak{osp}_{2n+1|2m}$ \cite{Palev}. 
 Throughout this note we will concentrate on the parafermionic algebra and its representations.

\section{Parafermionic Fock space}
\label{sec:3}
 The parafermionic relations (\ref{1}) imply that
the generators $E_{i}^j=\demi [a_i^{\dagger},a^j]$ 
are the matrix units satisfying
\[
[E_{i}^j, E_{k}^l]= \delta_{k}^j E_{i}^l - \delta_{i}^l E_{k}^j \ .
\] 
These generators  close the
 real form $ \mathfrak{u}$ of  a linear algebra $\mathfrak{gl}_n$ with $(E_i^j)^{\dagger}=   E_{j}^i $.
 
 One  has decomposition of the parafermionic  Lie algebra into reductive algebra $\mathfrak{u}$ and nilpotent Lie algebras, $\n$ and $\n^\ast$
 $$\g  = \n^{\ast} \rtimes \mathfrak{u} \ltimes \n \ $$
 where $\mathfrak{u}$ is the real form of the linear algebra $\mathfrak{gl}_n$.
 The free 2-step nilpotent Lie subalgebra $\n\subset \g$ is generated in degree 1 by the {\it creation} operators $a_i^{\dagger}$, $V:=\bigoplus_i \C
 a_i^\dagger$
 $$ \n=\n_1\oplus \n_2= V \oplus \wedge^2 V 
% \qquad \qquad (\n^\ast= V^\ast \oplus \wedge^2 V^\ast) 
 \ .
 $$
 % The algebra $\n$ is positively graded with degrees $\n_1=\bigoplus_i \C
 %a_i^\dagger=:V$ and  $\n_2=[\n_1, \n_1]
 %\cong \wedge^2 V^\ast$ ( and $\n_i^\ast=0$ for $i\geq 3$).
Analogously the  annihilation operators $a_i$ generate
the subalgebra $\n^\ast= V^\ast \oplus 
 \wedge^2 V^\ast$. %$\n^{\ast}\subset \g$,
% $\n^{\ast}= \n^\dagger$.

 The vector space $V=\n_1$  is the  fundamental representation  for  the left action of the algebra $\mathfrak{gl}_n$,
 $E_{i}^j \cdot a_k^{\dagger}= \delta^j_{k} a_i^{\dagger}$. Similarly $V^\ast=\n_1^\ast$ is the fundamental
 representation for the  right  $\mathfrak{gl}_n$-action,
 $  a^k \cdot E_{i}^j= \delta_i^{k} a^j$. 
The %real form  $\mathfrak{u}$ of 
 linear algebra $\mathfrak{gl}_n$ acts on the algebras $\n$ and $\n^\ast$  by automorphisms.
%Let us consider an unique  vacuum state $\rvac$ such that it is annihilated by
%all annihilation operators $a^i$.

\begin{definition}
The parafermionic Fock  space  is
the unitary representation $\f$ of the parafermionic algebra
$\g \cong {\mathfrak{so}_{2n+1}}$
built on a unique vacuum vector $\rvac$ such that 
\beq
a_{i}\rvac =0 \ ,  \qquad 
[ a_{i},a_{j}^{\dagger} ]\rvac =p \delta_{ij} \rvac \ .
\eeq
The non-negative integer $p$ is called the order of the parastatistics.
\end{definition}

Let us single out a particular parabolic subalgebra $\mathfrak{p}=  \mathfrak{gl} \ltimes \n$. In the Fock representation the
vacuum module $\C\rvac $ is the trivial module for 
the subalgebra  $\mathfrak{p}^\ast= \n^\ast \rtimes \mathfrak{gl}$.
The  representation induced by $\mathfrak{p}^\ast$
acting on the vacuum module is isomorphic
the  universal enveloping algebra
%\footnote{The algebra $U\n$ was denoted $PS(V)$ in the paper \cite{LP} 
%as an abriviation of ParaStatistics} 
of the creation algebra $\n$
 $${\rm{Ind}}_{\p^\ast}^\g \C \rvac = U \g \otimes_{\p^\ast} \C \rvac \cong U \n \ .$$ 
Hence the Fock representation $\f$ which we now describe is a particular quotient of the  algebra $U \n$ created by the free action of the creation algebra $\n$. % acting on the vacuum $\rvac$.

The   $\f$ of parastatistics order $p$
 is a finite-dimensional $\g$-module with a unique Lowest  Weight vector $\rvac$ of weight $-\frac{p}{2} \sum_{i=1}^n e_i$
 and  a unique Highest Weight (HW) vector
 \beq
 \label{hwv}
  |\Lambda\rangle = (a_1^{\dagger })^p \dots (a_n^{\dagger})^p \rvac 
  %\, \,( \cong \langle 0 |^\ast )
 \eeq
    thus the $\mathfrak{so}_{2n+1}$-module $\f$  is a highest weight module of weight 
  $\Lambda$ %=\frac{p}{2} \sum \epsilon_i$
 $$
  V^{\Lambda}= \f  \qquad \qquad \Lambda=\frac{p}{2} \sum_{i=1}^n e_i\ .
 $$
 The parafermionic algebra of order $p=1$ coincides with the canonical fermionic Fock space, i.e., the HW representation  $\mathcal{V}(1)=V^\theta$ with $\theta=\frac{1}{2} \sum_{i=1}^n e_i$. 
The physical meaning of the order $p$ for the parafermionic algebra is the number of particles
that can occupyone and the  same state, that is, we deal with a  Pauli exclusion  principle
of order $p$. 
The  symmetric submodule $S^{p+1} \n_1 \subset \n_1^{\otimes p+1}$
 is spanned by the ``exclusion condition'' 
% $a_i^{\dagger p+1}=0$ 
 $(a_i^{\dagger })^{p+1}=0$
  and it generates an ideal $(S^{p+1} \n_1)$. 
  The parafermionic Fock space $\f$ %V^\Lambda$ 
is a Lowest  Weight module % of given $p$ 
 isomorphic to   the factor module  of $U \n$ % the Verma module $PS(V)$ 
by the ``exclusion'' ideal $(S^{p+1} \n_1)$
$$\f \cong%{\rm{Ind}}_\mathfrak{p}^\g \C\rvac/I(p)= PS/I(p) 
 U\n / (S^{p+1}\n_1)%\cong (S(V)\otimes S(\wedge^2 V))/(S^{p+1}V) % = PS(V)/(S^{p+1}V) 
 \ .
 $$
  On the other hand the parafermionic Fock space $\f=V^\Lambda$  is a HW $\g$-module with HW vector $| \Lambda \rangle
  $ %=(a_1^{\dagger })^p \dots (a_n^{\dagger})^p \rvac $
   (\ref{hwv})
 $$ V^{\Lambda} \cong
 U\n^\ast / (S^{p+1}\n_1^\ast) ) = \f \ .$$
 
\begin{theorem}[A.J. Bracken, H.S. Green\cite{BH}]
The HW $\mathfrak{so}_{2n+1}$-module $V^{\Lambda}\cong \f$ of HW vector $| \Lambda \rangle=| p \theta \rangle$
 splits into a sum of  irreducible $\mathfrak{gl}_n$-modules  $V^{\lambda}$
\beq
\label{branch}
V^{\Lambda}\downarrow^{\mathfrak{so}_{2n+1}}_{\mathfrak{gl}_n}
=\bigoplus_{\lambda: \lambda\subseteq (p^n)} V^{\lambda -(p/2)^n} 
\ , \qquad \qquad \Lambda =% p \theta=
\frac{p}{2} \sum_{i=1}^n e_i  \eeq
where the sum runs over all partitions which match inside the 
Young diagram $(p^n)$.% $(p,p,\ldots ,p)$.% $n\times p$. %, e.g.
%$\yng(p,p,\ldots ,p)  $
%\qquad n=5 \quad p=3 .
%$$
\end{theorem}
\begin{proof} The Weyl character formula applied to a Schur module $V^{\lambda}$
%and Weyl group $W=S_n$
 yields
the  Schur polynomial
$$
s_{\lambda}(x_1, \ldots ,x_n)={\sum_{w\in W_1} \varepsilon(w) e^{w(\rho_1+ \lambda)}}/
\sum_{w\in W_1} \varepsilon(w) e^{w(\rho_1)} 
\qquad  W_1:= S_n \ ,
$$
where the variables are $x_i := \exp({-e_i})$
and the vector $\rho_1=\demi\sum_{i=1}^n(n-2i+1)e_i$. Alternatively the Schur polynomial is written as a quotient of determinants
\beq
\label{schurp}
s_{\lambda}(x_1, \ldots, n_n)= \frac{\det||x^{\rho_{1i}+\lambda_i}_j||}{
\det||x^{\rho_{1i}}_j ||} \ .
\eeq

The Weyl character formula applied to the ${\mathfrak{so}_{2n+1}}$-module $V^{\Lambda}$
reads %\footnote{The identification $x_i := \exp({e_i})$}
\beq
\label{charMac}
\chi^{\Lambda}= D_{\rho  + p \theta} / D_{\rho}=
e^{ p \theta }\sum_{\lambda:\, l(\lambda')\leq p} s_{\lambda}(x_1, \ldots ,x_n) \ , \qquad e^{ p \theta }=(x_1 \ldots x_n)^{-\frac{p}{2}}
 %\quad \Lambda =  p\theta %=\demi \sum e_i
\eeq
where   $W=S_n\ltimes \Z^n_2$ is the Weyl group of the root system of Dynkin type $B_n$ and $D_{\rho}= \sum_{w\in W} \varepsilon(w) e^{w\rho}$ with 
$\rho=\demi\sum_{i=1}^n(2n-2i+1)e_i$. The  quotient of determinants $D_{\rho  + p \theta} / D_{\rho}$ can be further expanded 
as a sum over the Schur polynomials with no more than $p$ columns 
% can be represented as 
(see p.84 in the book of Macdonald \cite{Macdonald}).
Here $\lambda'$ stands for the partition conjugated to $\lambda$
and $l(\mu)$ is the length of the partition $\mu$.
The Schur polynomials $s_{\lambda}(x)$ are characters of the $\mathfrak{gl}_n$-modules
thus the expansion of the ${\mathfrak{so}_{2n+1}}$-character $\chi^{\Lambda}$ implies
the branching formula (\ref{branch}). We are done. \qed
\end{proof}

\section{Kostant's theorem and the cohomology $H^{\bullet}(\n, \f)$}

The Kostant theorem is a powerfull tool helping to calculate 
 cohomologies.
Let's have a semi-simple algebra $\g$ and its Borel subalgebra $\mathfrak b= \h \oplus \bigoplus_{\alpha\in \Delta_+} \g_\alpha \ .$
Any parabolic subalgebra $\p$,  $\g \supset \p \supseteq \mathfrak b$ 
has a Levi decomposition $\p = \g_1 \ltimes \n$ where
$\g_1$ is a reductive algebra and $\n$ is the nilradical (largest nilpotent ideal) of $\p$. Consider the $\g$-module $V^\Lambda$ of weight $\Lambda$  and the cohomology
$H^{\bullet}(\n, V^\Lambda)$ with coefficients
in the restriction $\n$-module $V^\Lambda\downarrow_{\n}^\g$.
The Kostant's theorem gives the decomposition of 
$H^{\bullet}(\n, V^\Lambda)$
as a sum of irreducibles $\g_1$-modules $V^{\mu}$.

\begin{theorem}(Kostant)
Let $W$ be the Weyl group of the algebra $\g$ and
the subset  $\Phi_{\sigma}\subseteq \Delta_+$ be
$$\Phi_{\sigma} :=\sigma  \Delta_- \cap\Delta_+ \subseteq  \Delta_+\ .$$
Let $\rho$ be the Weyl vector $\rho=\demi \sum_{\alpha \in \Delta_+} \alpha $.
  The roots of the nilpotent radical $\n$ are denoted as  $\Delta(\n)$ and
the subset $W^1=\{\sigma\in W| \Phi_\sigma\subset \Delta(\n)\}$
 is a cross section of the coset $W_1\backslash W$.
The cohomology $H^{\bullet}(\n, V^\Lambda)$ has a decomposition into irreducible $\g_1$-modules $V^{\mu}$
\[ 
H^{\bullet}(\n, V^\Lambda)= \bigoplus_{\sigma \in W^1}  V^{\sigma(\rho +\Lambda) - \rho}
\]
where 
 the cohomological degree of $H^{j}(\n)$ is  the number of the elements $j:=\#\Phi_\sigma$.
\end{theorem}
 
% \subsection{The cohomology $H^{\bullet}(\n, \f)$}
 J. Grassberger, A. King and P. Tirao \cite{tirao} applied Kostant's theorem to cohomology  $H^{\bullet}(\n,\C)$ with trivial coefficients. 
 Here we extend their method for cohomologies with coefficients in  the parafer\-mio\-nic Fock space $\f$, $H^{\bullet}(\n,\f)$.
 
  \begin{theorem}
  \label{c1}
  Let $\n$ be the free 2-step nilpotent Lie algebra
$\n= V\oplus \wedge^2 V$ and  $V^{\Lambda}$ be the
parafermionic Fock space, $V^{\Lambda}=\f$ .
  The cohomology $H^{\bullet}(\n, V^{\Lambda})$ with values in the
$\n$-module $V^{\Lambda}\downarrow^\g_\n$ has a decomposition into irreducible $\mathfrak{gl}(V)$-modules
 \beq
 \label{Hkp}
H^{k}(\n, \f) \cong \bigoplus_{\mu:\mu=\mu'} V^  {\ast \mu^{(p)}-(\frac{p}{2})^n} ,
\qquad \qquad 
k = \frac{1}{2}(|\mu| + r(\mu)) \ ,
\eeq
where the sum is over self-conjugated Young diagrams $\mu=(\alpha|\alpha)$
  and the notation $\mu^{(p)}$ stays for the $p$-augmented diagram 
  $\mu^{(p)}=(\alpha+p |\alpha)$.
  \end{theorem}
  We recall the Frobenius notation  for a Young diagram $\eta$
%$$\eta:=\left(\ba{ccc} \alpha_1 , \ldots, \alpha_r \\
%\beta_1 ,  \ldots, \beta_r \ea\right)
% $$ 
 $$\eta:=( \alpha_1 , \ldots, \alpha_r |
\beta_1 ,  \ldots, \beta_r ) \qquad r= r(\eta)
 $$ 
 where the {\it rank} $r(\eta)$ is the number of boxes on the diagonal of $\eta$,
the arm-length  $\alpha_i$ is the number of boxes
on the right of the $i$th diagonal box, and the leg-length $\beta_i$ is the number
of boxes below the $i$th diagonal box. 
The overall number of boxes  in   $\eta$ is
$
|\eta| = r+  \sum_{i=1}^{r} \alpha_i + \sum_{i=1}^r \beta_i  \ .
$
The conjugated diagram
$\eta'$ is the diagram in which the arms and legs are exchanged
%$$ \eta':=\left(\ba{ccc} 
%\beta_1 ,  \ldots, \beta_r \\ \alpha_1 , \ldots, \alpha_r  \ea\right) \ .$$
$$ \eta':=(
\beta_1 ,  \ldots, \beta_r | \alpha_1 , \ldots, \alpha_r ) \ .$$
%It is worth noting that 
% the set of self-conjugated Young diagrams  $%\Delta^{\rho}=
% \{\lambda| \lambda=\lambda' \}$
%is singled out by the condition  $\beta_i=\alpha_i$ in Frobenius notation.
\begin{proof}
The parafermionic algebra $\g\cong\mathfrak{so}_{2n+1}$ has Cartan decomposition (\ref{bn}). Consider its parabolic subalgebra 
$
\p = \bigoplus_{i>j} \g_{e_i-e_j} \oplus \h \oplus \bigoplus_{\alpha\in \Delta_{+}} \g_{\alpha} \subset \g
$.
From the parafermionic relations (\ref{1}) is readily seen that
the Levi decomposition of the parabolic subalgebra $\p=\g_1 \ltimes \n$ 
has reductive component 
\beqa
\g_1= \h \oplus \bigoplus_{i\neq j} \g_{e_i-e_j} \ \cong \mathfrak{gl}_n 
\eeqa
acting by automorphisms on the free 2-step nilpotent algebra $\n$ (the space $\n_1=V$ being the fundamental representation of $\g_1=\mathfrak{gl}_n$) 
\beqa
\n&=&\bigoplus_{i}\g_{e_i} \oplus \bigoplus_{i< j} \g_{e_i+e_j} \cong V \oplus
\wedge^2 V \ .
\eeqa

The Weyl group $W_1$ of $\g_1=\mathfrak{gl}_n$ is the symmetric group $S_n$
operating on $\{e_1, \ldots ,e_n\}$ by permutations.
The 
Weyl group of $\g=\mathfrak{so}_{2n+1}$ is $W=S_n
\ltimes \Z^n_2$. The  $\Z^n_2$ 
is generated by operators $\tau_i$, $i=1,\ldots,n$ such that $\tau_i^2=1$ acting by
$$
\tau_i (e_j) =
\left\{ \ba{rcr} -e_j && i=j \\
e_j && i\neq j  \ea \right. \ .
$$
The elements $\tau_I\in \Z^n_2$  are indexed by subsets $I\subseteq \{1, \ldots, n\}$, $\tau_I \in \prod_{i\in I} \tau_i$.

Let us describe the subset $W^1$ which has order $|W^1|=2^n$.
Both $W^1$ and $\Z^n_2$ are cross sections of  $W_1\backslash W$
thus for each $\tau_I\in \Z_2^n$ exists a unique  permutation $\omega_I \in S_n$
such that $\omega_I \tau_I \in W^1$.

Let $\mathfrak b^{0}$ be the nilpotent part of the Borel algebra
$\mathfrak b^{0}=\mathfrak b \slash \h$ and and the complement be
$\mathfrak m_1=\g_1 \cap \mathfrak b^{0}=\mathfrak b^{0} / \n$.
The subset $W^1=\{ \sigma\in W| \Phi_\sigma \subseteq \Delta(\n) \}$
keeps stable also the complement of $\Delta(\n)$ 
$$
\sigma\Delta(\n)\subseteq \Delta_+ \qquad \Leftrightarrow \qquad
\sigma^{-1} \Delta(\mathfrak b^{0} / \n)\subseteq \Delta_+ \ .
$$
The root system of $\mathfrak m_1$ is  $\Delta(\mathfrak m_1)=\{e_i- e_j, i<j \}$
therefore $\omega_I \tau_I \in W^1$ implies $\tau^{-1}_I\omega^{-1}_I\Delta(\mathfrak m_1)\subseteq \Delta_+$ or
$ \tau_I\omega^{-1}_I (e_i- e_j)>0 $ for  $ i<j \ .$
These inequalies are satisfied for   $\omega_I\in S_n$ defined by 
$$\omega_I(a)>\omega_I(b) \quad \mbox{when} \quad \left\{ \ba{cccc} a<b& & a\in I & b\in I \\ a> b &&a\notin I &b\notin I\\ && a\in I& b\notin I\ea \right. \ .$$
The permutation places all elements of $I=\{i_1, \ldots i_r\}$  after all the elements of its complement
$\bar{I}$ preserving the order of $\bar{I}$ and reversing the order of 
$I$, that is,
\beq\label{oI}
\omega_I(1, \ldots , i_1,  \ldots ,  i_r, \ldots, n)=
(1, \ldots , \hat{i}_1,  \ldots ,  \hat{i}_r, \ldots, n, i_r, \ldots, i_2,i_1) \ . 
\eeq
The permutation $\omega_I$ can be  represented as  a product of 
cyclic permutations 
%$$\omega_I= \zeta_{n-i_k +k} \ldots \zeta_{{n-i_2+2}}\zeta_{n-i_1+1}$$
$\omega_I= \zeta_{i_r} \ldots \zeta_{{i_2}}\zeta_{i_1}$
where  $\zeta_{i_k}$ is the  cycle (of length ${n-i_k+1}$) from
positions $i_k-k+1$ to $n-k +1$.
Therefore the action of $\omega_I$ is represented by the sequence of steps
\beqa \zeta_{i_1 } (1, \ldots , i_1, \ldots, i_k,\dots n) &=&
 (1, \ldots ,\hat{i}_1, i_1+1,  \ldots,n, i_1) , \nonumber \\
 \zeta_{i_2 }(1,  \dots ,\! \!\!\!\!\! \!\underbrace{ i_2}_{\mbox {place } i_2-1 } \!\!\!\!\! , \ldots,n, i_1)&=&
 (1,  \dots ,\hat{i}_2,  \ldots,n,i_2, i_1), \nonumber\\
& \ldots& \nonumber \\
 \zeta_{i_k }(1,  \dots ,\! \!\!\!\!\! \!\underbrace{ i_k}_{\mbox {place } i_k-k+1 } \!\!\!\!\! , \ldots,n,i_{k-1},\ldots, i_1)&=&
 (1,  \dots ,\hat{i}_k,  \ldots,n,i_k, \ldots,  i_1) \ . \nonumber
\eeqa
Note that after the $j$-th step, the last $j$ places are not touched by
the next cyclings.

The Weyl vector $\rho$ associated to $\g=\mathfrak{so}_{2n+1}$ reads $\rho=\demi\sum_{i=1}^n(2n-2i+1)e_i$. Note that the components of $\rho$
are strictly decreasing with step $1=\rho_{i+1} - \rho_{i}$.
%The Kostant's theorem states that 
The cohomology ring $H^{\bullet}(\n,V^{\Lambda})$ decomposes into
 $\mathfrak{gl}(V)$-modules with HW weights $\sigma(\rho+ \Lambda) - \rho$ for $\sigma\in W^1$.
 We are interested in the case  $\Lambda =\frac{p}{2}\sum e_i $,
 $V^{\Lambda}=\f$.

Consider first the case $p=0$, i.e., the cohomology with trivial coefficients $H^{\bullet}(\n, \C)$ following \cite{tirao}.
The highest weights $\lambda_I= \sigma(\rho)- \rho$ for $\sigma\in W^1$ 
 are non-positive   due to
 $\sigma(\rho)_i\leq \rho_i$. 
 The cycling structure of $\omega_I$ implies
%  the following form of   $\lambda_I=\sum \lambda_j e_j$
 $$
 \lambda_I=\sum \lambda_j e_j, \qquad 
 \lambda_j ={- (n - i_{n-j+1} +1)}%_{-\rho_{i_k-k+1} - \rho_{n-k+1}}
  \chi_{(n-r+1\leq j \leq n)}- \sum_{k=1}^{r} \chi_{(i_k-k+1\leq j \leq n-k)} \ .
 $$
 One has an  isomorphism between a HW $\mathfrak{gl}_n$-module $V^{\lambda_I}$ with negative weight $\lambda_I\leq 0$ and the dual representation $V^{\ast \mu_I}$
 with reflected weight $\mu_I\geq 0$% \cite{FH}
$$
V^{\lambda_I} \cong V^{\ast \mu_I} \qquad \qquad \mu_I := \sum_{i=1}^n \mu_i e_i = -\sum_{i=1}^n  \lambda_{n-i+1} e_i \geq 0 \ . % \qquad \mu_1\geq \ldots\geq  \mu_n\geq 0  \ .
$$ 
The components of $\mu_I$ are decreasing positive integers
$\mu_1\geq \ldots\geq  \mu_n\geq 0$
\beq
\label{dom}
 \mu_j = (n - i_{j} +1) \chi_{(1\leq j \leq r)}+ \sum_{k=1}^{r} \chi_{(k+1\leq j \leq n- i_k +k)}  \ ,
\eeq
 and  these components  code a self-conjugated Young diagram $\mu^{\prime}_I=\mu_I$
 $$
 \mu_I= (\alpha_I|\alpha_I) \qquad \alpha_I=(\alpha_1, \ldots, \alpha_r), \quad
 \mbox{for} \quad
 \alpha_j=n-i_j \ .$$
Roughly speaking  the $j$-th cyclic permutation $\zeta_{i_k}$ in $\omega_I$ 
creates a self-conjugated  hook subdiagram of $\mu_I$ 
with  $\alpha_j =n-i_j$.

By virtue of the Kostant's theorem \cite{Kostant} the cohomology  $H^\bullet(\n, \C)$ of the  nilpotent Lie algebra  $\n$ 
has decomposition into Schur modules with HW vector $| \mu_I\rangle$
 $$H^\bullet(\n, \C)=\bigoplus_{\mu_I: \mu'_I=\mu_I}%\subseteq (n)^n} 
 V^{\ast \mu_I},
 \qquad  %\Omega_n= \qquad
  | \mu_I\rangle =E^{-\Phi_{\sigma}}, \quad \sigma\in W^1$$
 labelled by self-conjugated Young  diagrams. All self-conjugated Young  diagrams $\{\mu_I: \mu'_I=\mu_I\}$ are in bijection with elements of  $W^1$ (with cardinality $|W^1| = 2^n$), all these diagrams  are included into
the maximal square diagram, $\mu_I \subseteq (n^n)$.

Consider now the cohomology ring $H^{\bullet}(\n, V^{\Lambda})$ where $\Lambda =\frac{p}{2}\sum e_i$. It decomposes into
 $\mathfrak{gl}_n$-modules with  HW weights $\lambda_I^{(p)}=\sigma(\rho+ \Lambda) - \rho$ where $\sigma=\omega_I \tau_I \in W^1$.
% is again given by eq. (\ref{oI}).
Given a set  $ I=\{i_1, \ldots , i_r\}$  the shift $\Lambda$ modifies 
  the dominant weight $\nu_I= \sum \nu_i e_i$  to % as follows
  \[
\nu_j^{(p)}= - \lambda_{n-j+1}^{(p)} , \qquad  \nu_j^{(p)} =-\frac{p}{2}+ (n - i_{j} +1+p) \chi_{(1\leq j \leq r)}+ \sum_{k=1}^{r} \chi_{(k+1\leq j \leq n- i_k +k)}  \ .
\]
The weights $\nu_I^{(p)}=  \mu_I^{(p)}- \frac{p}{2}\sum  e_i$ fix  the HW vectors in  the $\mathfrak{gl}_n$-modules $V^{\ast\nu_I^{(p)}}$ %$V^{\ast \nu_I}:= V^{\ast \mu_I^{(p)}}$
$$V^{\ast\nu_I^{(p)}} =V^{\ast \mu^{(p)}_I} \otimes | \Lambda \rangle\quad \mbox{where} 
\qquad \mu^{(p)}_I=(
\alpha_I +p |\alpha_I) \qquad \alpha_j=n-i_j $$
from where the decomposition of $H^{\bullet}(\n, \f)$ (\ref{Hkp}) follows,
the sum over $\sigma \in W^1$ in Kostant's theorem being replaced by the  sum over self-conjugated Young diagrams $\mu=\mu'$.
%where  $\mu^{(p)}$ is one of the diagrams $\mu^{(p)}_I=(
%\beta|\alpha)$  with arm-lengths $\beta_j(\mu_I^{(p)})=n-i_j+p$
%and leg-lengths $\alpha_j(\mu_I^{(p)})=n-i_j$, for $j=1, \ldots, r$.
The arm $p$-augmented diagram 
$\mu^{(p)}_I$ stems  from   the self-conjugated  diagram $\mu_I=(\alpha_I|\alpha_I)$ cf. eq.  (\ref{dom}) 
by augmenting the arm-lengths, $\mu^{(p)}_I=(
\alpha_I +p |\alpha_I)$. %(see eq. (\ref{paug})).

The cohomological degree $k$ of the elements in 
$V^{\ast \mu_I^{(p)}}\otimes \rvac
\subset H^{k}(\n, \f) $ do not depend on $p$
%(which is the same as in $V^{\ast \mu_I}$)
but only on
 %the cardinality of the set  $\# \Phi_{\sigma}$. %, $\sigma \in W^1$.
% Let
  $\sigma=\omega_I \tau_I \in W^1$ (or equivalently on $\mu_I$).
  In view of $
\Phi_{\sigma} = 
  \Delta_- \cap\sigma^{-1}\Delta_+ $
a root $\xi \in \Phi_{\sigma}\subseteq \Delta(\n) $ whenever $\sigma^{-1} \xi< 0$.
  But the set $\Delta(\n)$ is stable under permutations and  
  $\tau_I^{-1}= \tau_I$ thus
  \beqa
  \#\Phi_\sigma &=& \# \{ \xi \in \Delta(\n), \tau_I \xi <0 \} \nonumber \\
  &=&\#\{  \g_{e_i}, i\in I \} + \# \{ \g_{e_i+e_j}: i<j, i \in I\}\nonumber\\
  &=& \sum_{i\in I} (1 +n - i)= r +   \sum_{k=1}^r (n - i_k)
  = r + s=\deg{ \mu_I}%   \sum_{k=1}^r \alpha_{k}
  \nonumber \ .
  \eeqa
  Thus the cohomological degree $k=\deg  \mu_I%=\deg (|\mu_I\rangle)
  =\#\Phi_\sigma$ is the total degree $k=(r+s)$ of the bi-complex 
  $ \wedge^s (\wedge^2 V^{\ast}) \otimes \wedge^s V^{\ast}$.
%  $$
%  k=\deg{ \mu_I}= r+ s 
%  $$
  The number of boxes above the diagonal in $\mu_I$ is
  $s= \demi (|\mu_I|- r)$ so finally one gets
  $
 k= \deg { \mu_I} = \demi ( r(\mu_I)+ |\mu_I|) \ .
  $
%  In particular,
%  the degree % $k_{max}=N$ 
%  of the maximal square diagram $\Omega_n=(n^n)$ is
%   $\deg \Omega_n=n+ \frac{n(n-1)}{2}=\frac{n(n+1)}{2}  \ .$ 
We are done.   \qed
   \end{proof}

% \subsection{The cohomology $H^{\bullet}(\n, \C)$}
%
%
%From our prespective the cohomology   $H^{\bullet}(\n, \C)$ with trivial coefficients is  the special case of parafermionic Fock space with $p=0$ when $\f=\C$. So we start by review of this simpler situation.
%\begin{thm} Let $\n$ be the free 2-step nilpotent algebra
%$\n= V\oplus \wedge^2 V$.
% 
% \noindent
% We got the cohomology $H^{\bullet}(\n, \C)$ but also the homology $H_{\bullet}(\n, \C)$ thanks to the isomorphism (\ref{iso}).
%  The homology $H_{\bullet}(\n, \C)$ was first calculated by J\'osefiak and Weyman \cite{JW} and then independently using different method  by Sigg \cite{Sigg}.

\section{Resolution of $\f$}
A general result of  Henri Cartan \cite{Cartan} states that  every
positively graded $\mathcal A$-module $M$ of a  graded algebra $\mathcal A=\oplus_{n\geq 0} {\mathcal A}_n$ allows for a minimal projective resolution by projective $\mathcal A$-modules. Moreover the notions of a projective and a free module coincide in the graded category. Thus for every positively graded $\mathcal A$-module $M$ there exists a minimal  resolution by free 
$\mathcal A$-modules.

The universal enveloping  algebra $U \n$ is a graded associative algebra and the parafermionic Fock space $\f=V^{\Lambda}$ is a positively graded $U \n$-module. There exists \cite{Cartan} a minimal free resolution $P_\bullet=\bigoplus_{k=0}^{N} P_k$ of the right $U\n$-module $\f^{\ast}$
\beq
\label{freeres} 
0\rightarrow P_N\rightarrow \ldots \rightarrow P_{1} \rightarrow P_{0} \stackrel{%1\otimes \epsilon
}{\rightarrow} \f^{\ast} \rightarrow 0
\eeq
by  free right $U \n$-modules $P_k=    E_k\otimes U \n $. 
%Here  the map $\epsilon$ is the augmentation  of $U \n$. 
 We apply the functor 
$- \otimes_{U\n} \C$ on the complex $P_{\bullet}$, where $\C$ is the trivial $U\n$-module.
The minimality of the resolution  $P_{\bullet}$ implies \cite{Cartan} that the differentials of
the  complex $P_\bullet\otimes_{U \n} \C$ vanish.
Hence the multiplicity spaces $E_k$ coincide with the homologies% of $\n$
%with coefficients in the right module $\f^{\ast}$
 $$E_k\cong {\rm{Tor}}^{U\n}_{k}(\f^{\ast},\C) =H_k(\n,\f^{\ast})  
\qquad \Rightarrow \qquad E_k^\ast \cong H^k(\n, \f) \ ,$$
%The homology and cohomology of $\n$ being related by an
where we used the isomorphism %(\ref{cartaniso})
$H_k(\n, M)^{\ast}=H^{k}(\n, M^{\ast})$.
%we find %the vector space isomorphism 
%$$ E_k^\ast \cong H^k(\n, \f) \ .$$
Theorem \ref{c1} gives us the spaces $E_k \cong H^k(\n, \f)^\ast$ so
  we have constructed the minimal free resolution (\ref{freeres}).

\begin{theorem} The Euler-Poincare characteristic of
the free minimal resolution of the (dual of the)
 parafermionic Fock space $\f$ (\ref{freeres})
yields  the identity 
 %parafermionic sign-alternating identity %(\ref{conj})
 \beq
\label{conj}
\frac{\sum_{\mu: \mu=\mu'}  (-1)^{\demi(|\mu|+r(\mu)) }s_{\mu^{(p)}}(x)}{{\prod_{i} {(1-x_{i})} 
\prod_{i<j}{(1-x_{i}x_{j}} ) }}
 =\sum_{\lambda: l(\lambda')\leq p} s_\lambda(x) \ .
\eeq
\end{theorem}
%{\bf Proof.}
\begin{proof}
In general, the  mapping of modules of an algebra  into its Grothendieck ring
of characters
 is an example of Poincar\'e-Euler characteristic.
The free resolution (\ref{freeres}) is naturally a (reducible) $\mathfrak{gl}(V)$-module and the Schur functions (\ref{schurp}) span the
ring of $\mathfrak{gl}(V)$-characters.
All the homology of a resolution is concentrated in degree 0, hence
on the RHS of (\ref{conj}) stays the character of 
the self-conjugated\footnote{The self-conjugacy $\f\cong \f^\ast$
allows to switch between  $x_i:=exp(\pm e_i)$ without a conflict.} module $\f$(\ref{charMac}) 
 $$ch \f =ch \f^\ast = e^{-p \theta} \sum_{\lambda \subseteq (p^n)} s_\lambda(x) \qquad x_i:=exp(e_i) \ .$$
 From the Poincar\'e-Birkhof-Witt theorem follows that the
 character of $P_k$ reads % (see \cite{Chaturvedi, D-VP}) is
 $$
ch\, P_k= ch (E_k\otimes U\n) =  \frac{e^{-\Lambda} s_{\mu^{(p)}}(x)}{\prod_i (1- x_i) \prod_{i<j} (1-x_i x_j)} \ .
 $$
 Thus the alternating sum on the LHS comes from the 
 characters of the  $\mathfrak{gl}(V)$-modules $ E_k\otimes U\n$ taken with alternating signs
 corresponding to  the homological degree. The factor $e^{p \theta}=e^\Lambda$ accounting for  the weight of the HW vector $|\Lambda \rangle$ cancels which proves
 the  parafermionic sign-alternating
  identity (\ref{conj}). \qed
 \end{proof}
 {\bf Remark.}  
 The free minimal resolution  of the trivial module 
 $\C$ %\cong \mathcal V(p=0)$ 
 constructed by J\'ozefiak and Weyman \cite{JW} with the help of the
 the homologies $H_{k}(\n,\C)$ corresponds to
 the resolution $P_{\bullet}$ (\ref{freeres})
 of $\C \cong \mathcal V(p=0)$.
 
 %$\hfill \boxed{}$
% found in the work \cite{SVdJ} of Stoilova and Van der Jeugt.
 
 The parafermionic sign-alternating identity (\ref{conj}) was proposed by 
 Stoilova and Van der Jeugt in their study of parafermionic Fock space \cite{SVdJ}.
The  parabosonic Fock space 
  has been explored in \cite{LSVdJ} where
 the ``super-symmetric partner'' of the identities   (\ref{conj}) has been proposed (for a combinatorial proof see \cite{King})
 \beq
\label{conj'}
\frac{\sum_{\mu: \mu=\mu'}  (-1)^{\demi(|\mu|+r(\mu)) }s_{[\mu^{(p)}]^\prime}(x)}{{\prod_{i} {(1-x_{i})} 
\prod_{i<j}{(1-x_{i}x_{j}} ) }}
 =\sum_{\lambda: l(\lambda)\leq p} s_\lambda(x) \ .
\eeq
The parity functor $\Pi$ switches parafermionic {\it even} generators
to parabosonic {\it odd} generators, thus
$\g =\mathfrak{so}_{2n+1} \stackrel{\Pi}{\rightarrow} \tilde{\g}=\mathfrak{osp}_{1|2n}$. The effect of  $\Pi$ is
the passage $\lambda \stackrel{\Pi}{\rightarrow} \lambda^\prime$.
The  identity (\ref{conj'}) is rooted into a minimal free resolution of the
parabosonic Fock space $\tilde{\mathcal V}(p)=\Pi \f$
 by free $U\tilde{\n}$-modules of the nilpotent Lie super-algebra $\tilde{\n}\subset \tilde{\g}$.

 More generally, one can consider the parastatistics Fock space
 ${\mathcal V}_{n|m}(p)$ of the  parastatistics Lie super-algebra
 $\g_{n|m}:=\mathfrak{osp}_{2n+1|2m}$ with $n$ parafermionic and $m$ parabosonic modes. We conjecture that there exists a complex of free $U\n_{n|m}$-modules of the maximal  nilpotent Lie superalgebra $\n_{n|m}\subset \mathfrak{osp}_{2n+1|2m}$
 whose cohomology is  ${\mathcal V}_{n|m}(p)$.
% by free $U\n_{n|m}$-modules of the maximal  nilpotent Lie superalgebra $\n_{n|m}\subset \mathfrak{osp}_{2n+1|2m}$. 
Then the Euler-Poincare characteristics of such a complex will  yield one more identity (which was obtained by different method in \cite{LP}) %(generalizing 
%identities (\ref{conj}) and (\ref{conj'}))
 \[
%\label{identp}
 \frac{ \prod_{i<j \,,\,\hat{i}\neq \hat{j}} (1+x_{i}x_{j}) \sum_{\mu: \mu=\mu'} (-1)^{\frac{1}{2}(|\mu| + r(\mu))} hs_{\mu^{(p)}}(x)}
 { \prod_{i}(1-x_{i}) \prod_{i<j \,,\,\hat{i}=\hat{j}}{(1-x_{i}x_{j})}  }
=\sum_{\lambda:\, \lambda_1\leq { p}} hs_{\lambda}(x) \ . 
\]
Here % $hs_{\lambda}(x)$ stands for
 the $(n|m)$-hook Schur polynomial $hs_{\lambda}(x)$  is the 
  character of the irreducible $\mathfrak{gl}_{n|m}$-module $ V^{\lambda}$, $hs_{\lambda}(x)=ch\,  V^{\lambda} $. The non-trivial 
  $\mathfrak{gl}_{n|m}$-modules $V^{\lambda}$
   are labelled by diagrams $\lambda$ such that 
   $\lambda_{n+1}\leq m$.

\end{document}